\documentclass[11pt]{article}
\usepackage{graphicx}
\usepackage{hyperref}
\hypersetup{colorlinks,allcolors=black}
\usepackage{amsthm,amsmath}
\usepackage{enumerate}
\usepackage{amsfonts}
\usepackage{bbold}
\usepackage{xcolor}
\newtheorem{theorem}{Theorem}
\newtheorem{lemma}[theorem]{Lemma}

\newtheorem{definition}[theorem]{Definition}

\theoremstyle{definition}
\newtheorem*{xrem}{Remark}

\numberwithin{equation}{section}  
\numberwithin{theorem}{section}  

\newcommand{\CC}{\mathbb{C}}

\newcommand{\RR}{\mathbb{R}}

\title{Quantifying State Transfer Strength on Graphs with Involution}
\author{Yujia Shi, Gabor Lippner}

\begin{document}
\date{}
\maketitle

\section*{Abstract}
This paper discusses continuous-time quantum walks and asymptotic state transfer in graphs with an involution. By providing quantitative bounds on the eigenvectors of the Hamiltonian, it provides an approach to achieving high-fidelity state transfer by strategically selecting energy potentials based on the maximum degrees of the graphs. The study also involves an analysis of the time necessary for quantum transfer to occur.

%%%%%%%%%%%%%%%%%%%%
\section{Introduction}
%%%%%%%%%%%%%%%%%%%%

We study quantum transport phenomena in a network of spin particles with $XX$ coupling, and with magnetic fields applied to the nodes. The network is given as a simple and connected graph $G$ with vertex set $V$ and an edge set $E$. The magnetic fields are described by a $Q: V \rightarrow \mathbb{R}$ function. 

It is well-known (see e.g.~\cite{survey}) that node-to-node quantum transportation can be analyzed by restricting the system to its single-excitation subspace, where the solution of the Schr\"odinger equation becomes what is known as the continuous-time quantum walk:
\begin{align}
    \psi(t)= e^{i t H} \psi(0).
\end{align}
Here, the Hamiltonian $H = A_G + D_Q$ where $A_G$ is the adjacency matrix of the graph $G$ while $D_Q$ is a diagonal matrix containing the values of the function $Q$. The function $\psi(t) : V \to \CC$ described the state of the system at time $t$. The probability of the quantum walk getting from node $u$ to node $v$ at time $t$ is given by $p(t) = |\langle u|e^{itH}|v\rangle|^2$. 

In the ideal scenario, if $p(t)=1$ at some time $t$, then we say there is a \textbf{perfect state transfer} between vertices $u$ and $v$. However, achieving perfect state transfer demands strict conditions. Notably, C. Godsil demonstrated that \textbf{ratio condition} is necessary for the existence of such vertex pairs \cite{when_pst}. 

It has been observed before that large, equal, magnetic fields applied to $u$ and $v$ can lead to transfer strength close to 1. In physics, this phenomenon is referred to as quantum tunneling. Let us define the transfer fidelity of the network from to $u$ to $v$ by 
\begin{equation}\label{eq:fidelity_limit}
    F(Q) = F(G,u,v,Q) = \sup_{t\geq 0} p(t).
\end{equation} 

In the rest of the paper we assume that $Q(u)=Q(v)$ and $Q(x) = 0$ for all nodes $x \neq u,v$. By a slight abuse of notation we will refer to the value $Q(u)=Q(v)$ simply as $Q$.

Y. Lin, G. Lippner, and S-T. Yau~\cite{Tunneling_2012} showed that $\lim_{Q \to \infty} F(G,u,v,Q) = 1$ given that the graph exhibits certain local symmetries around $u$ and $v$. While they provide a complete characterization of when $F(Q) \to 1$, their result is not quantitative.

%
%\begin{definition}
%    Two vertices are \textbf{$k$-cospectral} if, for any $t\leq k$, the probabilities of return at time $t$ in a random walk are the same. The \textbf{cospectrality} of $u$ and $v$, denoted by $co(u,v)$ is the maximum number $m$ such that $u$ and $v$ are $m$-cospectral.
%\end{definition}
%\begin{theorem}\cite{Tunneling_2012}\label{thm:tunneling}
%    Assume vertices $u$ and $v$ have potential $Q$ that is much larger than the other vertices. Denote their distance in the graph by $d$. If $co(u,v)\geq d$ then
%    \begin{align}
%        \liminf_{Q\rightarrow \infty}{\sup_{t\geq 0} p(t)}=1.
%   \end{align}
%    This is called \textbf{asymptotic state transfer}. Additionally, if the graph has an involution that preserves $Q$, then the tunneling time is of order $Q^{\text{dist}(u,v)-1}$.
%\end{theorem}
% Later on, M. Kempton, G. Lippner, and S-T. Yau further showed the existence of $Q$ for any given small positive number $\epsilon$, ensuring $p(t)>1-\epsilon$ on a symmetric graph \cite{involution}. This kind of quantum walk, where the probability approaches 1, is often referred to as \textbf{pretty good state transfer}. 

Recently, C.M. van Bommel and S. Kirkland~\cite{path} quantified the convergence in \eqref{eq:fidelity_limit} for the case of the two endpoints of a path graph. 

\begin{theorem}[\cite{path}] Let $G$ be the path graph on at least 3 nodes and $u,v$ be its two endpoints. Then for $Q \geq \sqrt{2}$ 
\begin{equation}\label{eq:kirkland_bound}
     F(Q) > \left(\frac{Q^2-2}{Q^2}\right)^2 \approx 1 - \frac{4}{Q^2}\end{equation}
 
Furthermore this fidelity is reached within time $O(Q^{n-2})$
\end{theorem}

These bounds are derived from a careful asymptotic analysis of the difference between the two largest eigenvalues $\lambda_1-\lambda_2$ and their corresponding eigenvectors. 

An important feature of \eqref{eq:kirkland_bound} is that it doesn't depend on the length of the path, $n$. In this paper, we prove, in the same spirit, a general lower bound that holds for any graph with an involution, along with an assessment of the time required for quantum transfer to occur. 

\begin{theorem}\label{thm:main}
Let $G$ be a graph with an involution $\tau: V(G) \to V(G)$ and let $u = \tau(v) \in V(G)$ and maximum degree $m$. Then for $Q \geq m$ 
\[ F(Q) > 1- \frac{16 \sqrt{m+1}}{\sqrt{Q}}\]
and this fidelity is achieved within time $O(Q^{d(u,v)-1})$ where $d(v
_i,v_j)$ denotes the distance between two vertices in the graph. 
\end{theorem}
It turns out that to achieve the quantum state transfer from a vertex to its image with fidelity $1-\epsilon$, it is sufficient to choose the energy potential based on the maximum degree of the graph, and the potential is of order $O(\epsilon^{-2})$. This is in contrast to the energy level of $O(\epsilon^{-0.5})$ required on a path, as derived from the lower bound on fidelity given in \cite{path}. While the order of $Q$ is higher, this approach extends the application to any graph that has an involution. This result is valuable in determining the minimal energy consumption necessary to achieve a specific probability level. Remarkably, these bounds depend on the maximum degree of the graph and are independent of its size.

The paper is structured as follows: in Section~\ref{sec:prelim} we introduce the basic setup and all relevant notation for graphs with an involution and the corresponding Hamiltonians. In Section~\ref{sec:main} we prove the main result modulo estimates of the eigenvalues and eigenvectors. We provide these in Sections~\ref{se: lambda} and~\ref{se: phi} respectively. Finally in Section~\ref{se:t} we analyse the time required to achieve strong state transfer.

%%%%%%%%%%%%%%%%%%%
\section{Preliminaries}\label{sec:prelim}
%%%%%%%%%%%%%%%%%%

%%%%%
\subsection{Spectral Decomposition}
%%%%%

In the current setup, the Hamiltonian $H$ governing the continuous time quantum walk is the sum of the adjacency matrix of the graph and a diagonal matrix. In the matrix presented below, $Q: V \to \RR$ represents the energy potential at the vertices and the off-diagonal entry $H_{ij}=1$ if there is an edge between $v_i$ and $v_j$:
$$H=\begin{bmatrix}
Q(v_1) & \mathbb{1}(v_1\sim v_2) & ... & \mathbb{1}(v_1\sim v_n)\\
\mathbb{1}(v_1\sim v_2) & Q(v_2) & ... & \mathbb{1}(v_2\sim v_n)\\
\vdots\\
 \mathbb{1}(v_1\sim v_n)&...&...& Q(v_n)
\end{bmatrix}.$$
Since $H$ is a real symmetric $n\times n$ matrix, it has $n$ real eigenvalues, $\lambda_1\geq\lambda_2\geq...\geq\lambda_n$ and their corresponding eigenvectors $\varphi_1, \varphi_2,...,\varphi_n$ which together form an orthonormal basis. This spectral decomposition of $H$ helps us to compute the solution to the Schr\"odinger equation. Given a vector representing the initial state $\psi(0)$, it can be written as a linear combination of the orthonormal basis $\psi(0)=\sum_{j=1}^n c_j\varphi_j$. In the context of this problem, $\psi(0)$ is usually the characteristic vector of the starting vertex $u$, meaning the system starts at vertex $u$ with probability $1$. It immediately follows that 
\begin{align}
    \psi(t)=\sum_{j=1}^n c_j e^{it\lambda_j}\varphi_j.
\end{align}
If our objective is to achieve strong state transfer between vertices $u$ and $v$, we primarily aim to find a system such that $\psi(t)\cdot e_v$ has a significantly large squared norm 
\begin{align}\label{eq:p(t)}
    p(t)=|\sum_{j=1}^n\varphi_j(u) \varphi_j(v) e^{it\lambda_j}|^2
.\end{align} Here, we consider the eigenvector $\varphi_j$ as a function $\varphi_j: V\rightarrow\mathbb{R}$ that returns the corresponding entry of an vertex. Therefore, the main focus of our study revolves around investigating how the structure of the graph and the energy potential assigned to its vertices influence the eigenvalues and eigenvectors at the endpoints. 
% an intuition behind quantum tunneling in Theorem \ref{thm:tunneling} is given by Gershgorin Circle Theorem.
% \begin{theorem}\label{thm:circlethm}
%     (Gershgorin Circle Theorem)
%     Let $A$ be a $n\times n$ complex matrix, then every eigenvalue of $A$ lies in a Gershgorin disc, which is centered at a diagonal entry $a_{ii}$ and has radius equal to the sum of the absolute values of the off-diagonal entries in the same row $\sum_{i\neq j} |a_{ij}|$.
% \end{theorem}
For instance, if the energy potential $Q$ on vertices $u$ and $v$ is relatively large compared to the potential on other vertices, and $Q\gg\max_{x\in V}\deg(x)$, then there are two eigenvalues that are significantly larger than the rest $\lambda_1\geq\lambda_2 \gg \lambda_3\geq ...\geq\lambda_n$.

%%%%%
\subsection{Graphs with Involution}\label{se:gw/involution}
%%%%%

\begin{definition}
    Vertices $u$ and $w$ are \textbf{strongly cospectral} if for any eigenvector $\varphi$ of $H$, $\varphi(u)=\pm \varphi(w)$. They are \textbf{cospectral} if this holds for at least a given orthonormal basis of eigenvectors $\varphi_1,\dots, \varphi_n$.
\end{definition}
In many quantum walk studies, a common aspect explored is the presence of strong cospectrality \cite{PGST_2017} or, to some extent, cospectrality \cite{Tunneling_2012}, since strong cospectrality is a necessary condition for two vertices to have perfect state transfer. Consequently, if the graph structure already exhibits cospectrality, it becomes easier for us to show strong state transfer. A notable example of such a graph is one with involution.
\begin{definition}
$G$ is a \textbf{graph with involution} if there is a bijection from the set of vertices to itself $\sigma : V (G)\rightarrow V (G)$ which satisfies the following conditions,

\begin{itemize}
    \item $\sigma\circ\sigma(u)=u$
    \item if $u\sim v$ then $\sigma(u) \sim \sigma (v)$
    \item $\sigma$ preserves the potential on $V$, that is $Q(v)=Q(\sigma(v))$
\end{itemize}

\end{definition}

For simplicity, we denote $\sigma(v)$ by $v'$ in the rest of this paper. Let $S=\{ v\in V| v=v'\}$ be the set of fixed vertices. Select a vertex from each pair of $\{v,v'\}$ that are not fixed by $\sigma$, then we get a partition of the vertex set $V=N \cup \sigma N \cup S$. Denote the sizes of the subsets by  $s=|S|$ and $k=|N|$. This divides the Hamiltonian into a block matrix
\begin{align}
H=
    \begin{bmatrix}
                    H'   &  A_\sigma & A_S\\
                    A_\sigma   &  H' & A_S\\
                    A_S^T     &  A_S^T & H_S
    \end{bmatrix}.
\end{align}
The $k\times k$ matrix $H'$ and the $s\times s$ matrix $H_S$ are the Hamiltonians of the subgraphs induced by $N$ and $S$ respectively. The $k\times k$ matrix $A_\sigma$ contains all the edges between $N$ and $\sigma N$. The $k\times s$ matrix $A_S$ contains all the edges between $N$ and $S$.
\begin{lemma}\label{le:H}\cite{involution}
    Let graph $G$ be a graph with involution $\sigma$, and $Q$ be its potential function. The Hamiltonian of $G$ has eigenvalues $\lambda_1\geq\lambda_2\geq...\geq\lambda_n$ and corresponding eigenvectors $\varphi_1$, $\varphi_2$ ...  $\varphi_n$. Then every $\lambda_j$ is either the 
    eigenvalue of 
    \begin{align}
    H^+=\begin{bmatrix}
    H' + A_\sigma   &  A_S\\
    2A_S^T          &  H_S
        \end{bmatrix}
    \end{align}
    or the eigenvalue of $H^-=H'-A_\sigma$.
\end{lemma}
It can be easily verified that given an eigenvector of the reduced matrix $H^+ a=\lambda_j a$, $\varphi_j$ is in the form 
$\begin{bmatrix}
    a&
    a&
    b
\end{bmatrix}^T$. On the other hand, if $H^-c=\lambda_jc$, then $\lambda_j$ has corresponding eigenvector $\begin{bmatrix}
    c&
    -c&
    0
\end{bmatrix}^T$. Here $a$ and $c$ represents vectors in $\mathbb{R}^k$, and $b$ and $0$ are vectors in $\mathbb{R}^s$. 
\begin{lemma}
    Denote the set of eigenvalues of $H^+$ by $\pi^+$ and the set of eigenvalues of $H^-$ by $\pi^-$, then $\max_{\lambda_i}\pi^+>\max_{\lambda_i}\pi^-$. In the double-well case, that is when $Q(v)=Q(v')\gg Q(v_i)$ for the rest of vertices, $\lambda_1$ and $\lambda_2$ are the largest numbers inside $\pi^+$ and $\pi^-$ respectively, as long as $Q(v)$ is significantly greater than the maximum degree of the graph.
\end{lemma}
\begin{proof}
    Since the graph is connected, $H$ is an irreducible non-negative symmetric matrix. By the Perron-Frobenius theorem, $\lambda_1$ is strictly greater than $\lambda_2$, and the components of its corresponding eigenvector are all positive. This implies $\varphi_1$ is symmetric and $\lambda_1\in\pi^+$.
    
    Gershgorin's Circle theorem tells us that every eigenvalue is bounded by a disc centered at some $Q(v_i)$ with radius $\deg(v_i)$. When $Q(v)$ is large enough such there is no intersection between $[Q(v)-m, Q(v)+m]$ and $[Q(v_i)-m, Q(v_i)+m]$, both $H^+$ and $H^-$ have a large eigenvalue close to $Q(v)$. Therefore, the second largest eigenvalue $\lambda_2$ is in $\pi^-$.
\end{proof}
By applying these lemmas to the spectral decomposition \eqref{eq:p(t)}, we conclude that if $u=\sigma(v)$, then
\begin{align}\label{eq:sum_over_pi}
    \sum_{\lambda_j}\varphi_j(u) \varphi_j(v) e^{it\lambda_j}=\sum_{\lambda_j\in \pi^+} \varphi_j(u)^2 e^{it\lambda_j}- \sum_{\lambda_j\in \pi^-} \varphi_j(u)^2 e^{it\lambda_j}
.\end{align}
% The orthonormal basis $\{\varphi_i\}_{i=1,2 ...}$ implies that $
% \sum_{j=1}^{2n+s} \varphi_j(u)^2=1$.
 
Therefore, when the two largest eigenvalues are differed by $\frac{\pi+2k\pi}{t}$, the sum of the leading terms reaches its maximum $\varphi_1(u)^2+\varphi_2(u)^2$. In this scenario, we can ensure the probability is close to 1 by showing both $\varphi_1(u)^2$ and $\varphi_2(u)^2$ are close to $\frac{1}{2}$.
%%%%%%%%%%%%%%%
\section{Achieving State Transfer with High Probability}\label{sec:main}
%%%%%%%%%%%%%%%

% To achieve a specified level of probability, the system should have eigenfunctions that exhibit high fidelities at the endpoints, with corresponding eigenvalues significantly greater than the rest. 
As before, we assume that only the endpoints $v$ and $v'$ possess a large potential $Q(v)=Q(v')=Q$ while the remaining vertices have zero potential $Q(v_i)=0$ going forward. Theorem~\ref{thm:main} is a direct corollary of the following statement.

% $p(t)\geq (\frac{Q^2-2}{Q^2})^2$ on path.

\begin{theorem}\label{thm:main2}
    Let $G$ be a graph with involution and maximum degree $m$. For every small real number $\epsilon>0$, if the potential $Q$ on $v$ and $v'$ satisfies 
    \begin{align}
        Q\geq \frac{256\cdot(m+1)}{\epsilon^2}
    \end{align}
    
    then there exists time $t<\frac{\pi}{2}(Q+m)^{d(v,v')-1}$ such that the probability of state transfer from $v$ to $v'$ is no less than $1-\epsilon$.
\end{theorem}

\begin{proof}
    When $t=\frac{\pi}{\lambda_1-\lambda_2}$ then $|\varphi_1(v)^2 e^{it\lambda_1}+\varphi_2(v)^2 e^{it\lambda_2}|=\varphi_1(v)^2+\varphi_2(v)^2$. The bound on time $t$ is further discussed in section \ref{se:t}. In Theorem \ref{thm:lwbd_phi}, we will prove that both $\varphi_1(v)$ and $\varphi_2(v)$ have lower bound $\sqrt{\frac{1}{2}-\frac{m}{2Q^2}}-\sqrt{\frac{m+1}{Q-m-1}}$. Certainly, $Q$ should be large enough such that the bound is greater than $\frac{1}{2}$. Thus, the probability has a lower bound depending on $Q$ and the maximum degree $m$,

    \begin{align*}\label{eq:sd}
    p(t)&=|\varphi_1(v)^2+\varphi_2(v)^2+\sum_{j=3}\varphi_j(v) \varphi_j(v') e^{it\lambda_j}|^2\\
    &\geq (\varphi_1(v)^2+\varphi_2(v)^2-(1-\varphi_1(v)^2-\varphi_2(v)^2))^2\\
    &\geq (2\varphi_1(v)^2+2\varphi_2(v)^2-1)^2,
\end{align*}
and so
    \begin{align}
        p(t)\geq\left(4\cdot \left(\sqrt{\frac{1}{2}-\frac{m}{2Q^2}}-\sqrt{\frac{m+1}{Q-m-1}}\right)^2-1\right)^2.
    \end{align}
    To guarantee the probability exceeds $1-\epsilon$, a simple computation shows it is sufficient for $Q$ to satisfy
    \begin{align}
        Q\geq \frac{(m+1)(c+1)}{c}, \text{where } c=\frac{1}{2}\cdot\left(\frac{1}{2}-\frac{1+\sqrt{1-\epsilon}}{4}\right)^2
    \end{align}
\end{proof}

%%%%%
\subsection{Lower Bounds on $\lambda_1$ and $\lambda_2$}\label{se: lambda}
%%%%%

The purpose of this section is to derive new lower bounds on the two largest eigenvalues that improve upon the ones given by the Gershgorin Circle Theorem.
\begin{theorem}\label{thm:lambda_1}
    Label the vertices in $N$ as $v_1$, $v_2$,...,$v_{k}$ and $S$'s vertices as $v_{k+1}$,...,$v_{k+s}$.  Denote the degree of a vertex in $N\cup S$ by $\text{deg}_{N\cup S}(v_i)$. Let $y$ be a vector in $\mathbb{R}^{k+s}$ and $y_i=Q^{-\text{min}\{\text{dist}(v_1,v_i),\text{dist}(v_1,v_i')\}}$. Then 
    \begin{align}\label{eq:lambda1}
    \lambda_1 \geq Q + Q^{-1} \cdot \frac{1+\sum_i (\text{deg}_{N\cup S}(v_i)-1)y_i^2}{\sum_i y_i^2}
.\end{align}
\end{theorem}

\begin{proof}
According to Lemma \ref{le:H}, $\lambda_1$ is also the largest eigenvalue of $H^+$. The vector $y$ is a good approximation of the eigenvector because its entries decrease as the vertex moves further away from $v_1$. Therefore, by calculating the Rayleigh quotient on $H^+$ and $y$, one can expect to obtain a reasonably close lower bound on $\lambda_1$. Without loss of generality, we make $H^+$ symmetric by conjugation 
\begin{align}
    H^+=\begin{bmatrix}
H' + A_\sigma & \sqrt{2}A_S\\
\sqrt{2}A_S^T &  H_S
\end{bmatrix}
.\end{align}

Notice that summing the nonzero entries $H^+_{ij}y_j$ in the numerator of the quotient is equivalent to counting the edges connected to vertex $v_i$. Therefore, the question becomes finding the number of vertices adjacent to $v_i$ and are closer, further away, or of equal distance to $v_1$.

When $i=1$,
    \begin{align}
        y_1\sum_{i,j} H_{ij}^{+} y_j \geq
    Q + Q^{-1}\text{deg} (v_1)
.    \end{align}
    This is because any vertex $v_j$ adjacent to $v_1$ has $y_j=Q^{-1}$; counting all the nonzero coefficients in front of $y_j$ gives the degree of $v_1$.
    
When $i=2,3 ,..., k+s$,
    \begin{align}
        y_i\sum_{i,j} H_{ij}^{+} y_j
    \geq(Q+(\text{deg}_{N\cup S}(v_i)-1)Q^{-1})y_i^2
.    \end{align}
The coefficient in front of $y_iy_j$ is nonzero if and only if $y_j=y_i$, $y_j=Qy_i$, or $y_j=Q^{-1}y_i$. Moreover, since the graph is connected, for every vertex $v_i$, one can find a vertex adjacent to it and is in one of the shortest paths from $v_1$ to $v_i$. Hence, there must exist $y_j=Qy_i$. 

Every term contains $Qy_i^2$ and $(\text{deg}_{N\cup S}(v_i)-1)Q^{-1}y_i^2$ which immediately implies the result.
\end{proof}

Noticeably, the lower bound given in\cite{path} for a path of length $n$ is $Q+Q^{-1}+O(Q^{2-n})$, close to our result $Q+Q^{-1}+O(Q^{-1})$ for graphs with involution. 

The calculation for $\lambda_2$'s lower bound follows a similar approach. Recall that $\lambda_2$ is the largest eigenvalue of $H^-=H'-A_\sigma$. Due to the possibility of negative numbers in the matrix, some adjustments need to be made. Consider the vertex set excluding all fixed vertices $V \backslash S = N \cup \sigma N$. Fix an endpoint $v_1$ in $N$ and assign vertex $v_i$ to $N$ if $v_i$ is closer to $v_1$ than to $v_1'$. In case the distances are equal, $v_i$ can belong to either $N$ or $\sigma N$. Next, define a new vector $y$ in $\mathbb{R}^{k}$
\begin{equation}
  y_i =
    \begin{cases}
      Q^{-\text{dist}(v_1,v_i)} & \text{if $\text{dist}(v_1,v_i)<\text{dist}(v_1',v_i)$}\\
      0 & \text{if $\text{dist}(v_1,v_i)=\text{dist}(v_1,v_i')$}
    \end{cases}       
.\end{equation}
\begin{theorem}
Let $\text{deg}_{N\cup\sigma N}(v_i)$ denote the degree of vertex $v_i$ in $N\cup\sigma N$. Use the same setting in Theorem \ref{thm:lambda_1} and the vector $y$ defined above. The lower bound on $\lambda_2$ is given by
    \begin{align}\label{eq:2}
    \lambda_2&\geq  Q - \frac{1+\sum_{i\neq 1} (\text{deg}_{N\cup\sigma N}(v_i)-1)y_i^2}{\sum_{i} y_i^2}
.\end{align}    
\end{theorem}
\begin{proof}
Again by Rayleigh's inequality,
\begin{align}
    \lambda_2\geq \frac{\sum_{i} y_i \sum_{i,j} H_{ij}^{-} y_j}{\sum_{i} y_i^2}
.\end{align}

When $i=1$,
    \begin{align}
        y_1\sum_{j} H_{1j}^{-} y_j 
        \geq Q - 1
    .\end{align}
This is because starting from $j=2$, $\mathbb{1}(v_1\sim v_j) - \mathbb{1}(v_1\sim v_j')$ has to be non-negative. If that was not the case, then $v_j$ would be closer to $v_1'$ instead of $v_1$, which would contradict the assumption.

When $i= 2,3,...,k$, similar to the proof of the previous theorem, there exists a vertex in $N\cup \sigma N$ that is adjacent to $v_i$ and is in one of the shortest paths from $v_1$ to $v_i$. Consequently, there must exist a $y_j=Qy_i$. There are three cases for the coefficient in front of $y_iy_j$. If $[\mathbb{1}(v_1\sim v_j) - \mathbb{1}(v_1\sim v_j')]$ takes the form $(1-1)$ or $(0-1)$, it is easy to verify that $\text{dist}(v_1,v_i)=\text{dist}(v_1,v_i')$, which implies  $y_iy_j$ is zero. 

Sum over $y_iy_j$ only when its coefficient is $(1-0)$, and it follows that
    \begin{equation}
        y_i\sum_{j} H_{i,j}^{-} y_j  \geq Qy^2_i - (\text{deg}_{N\cup \sigma N}(v_i)-1)y_i^2
.\end{equation}
\end{proof}
Notice that this lower bound is greater than $Q-\max_{v_i\in V}{\text{deg}(v_i)}$, which is an improved bound compared to the one given by the Circle Theorem.

%%%%%
\subsection{Lower Bounds on $\varphi_1(v_1)$ and $\varphi_2(v_1)$}\label{se: phi}
%%%%%

Recall the probability of quantum state transfer from vertex $v_1$ to vertex $v_1'$ at time $t$ can be written as the following expression
\begin{align}\label{eq:p(t)_pi}
    p(t)=\left|\sum_{\lambda_j\in \pi^+} \varphi_j(v_1)^2 e^{it\lambda_j}- \sum_{\lambda_j\in \pi^-} \varphi_j(v_1)^2 e^{it\lambda_j}\right|^2
\end{align}
where eigenvectors $\varphi_j$ form an orthonormal basis. 
To demonstrate that this probability can be close to 1 at some time $t$, it suffices to prove that both $\varphi_1(v_1)$ and $\varphi_2(v_1)$ are close to $\frac{1}{\sqrt{2}}$.
\begin{theorem}\label{thm:lwbd_phi}
    Let $m$ be the maximum degree of the graph. If $Q$ is greater than $2m$, then the corresponding normalized eigenvectors of $\lambda_1$ and $\lambda_2$ satisfy
    \begin{align}
        \varphi_1(v_1)\geq \sqrt{0.5-\frac{m}{2Q^2}}-\sqrt{\frac{m}{Q-m}}\\
        \varphi_2(v_1)\geq \sqrt{0.5-\frac{m}{2Q^2}}-\sqrt{\frac{m+1}{Q-m-1}}
.    \end{align}
    
\end{theorem}
\begin{proof}
Consider the vector $y$ we constructed earlier in Section \ref{se: lambda}, but this time in $\mathbb{R}^{2k+s}$. It can be expressed as $y=\sum^{2k+s}c_i \varphi_i$, a linear combination of the eigenvectors. Denote the squared Euclidean norm of $y$ by $D$. Due to the involution, $y(v)=y(v')$ is true for every vertex. If $\lambda_i$ is in $\pi^-$ then its corresponding eigenvector alternates, meaning $\varphi_i(v)=-\varphi_i(v')$. On the other hand, $\varphi_i(v)=\varphi_i(v')$ if $\lambda_i$ is in $\pi^+$. This implies $c_i$  is 0 for every $\lambda_i\in\pi^-$. By Cauchy-Schwartz inequality, 
\begin{equation}\label{eq:4}
    \varphi_1 (v_1)\geq \frac{1-\sqrt{D-c_1^2}}{\sqrt{D}}
.\end{equation}
Denote the difference between $\lambda_1$ and the Rayleigh quotient on $H^+$ and $y$ by $\epsilon$, 
\begin{align}
    \epsilon=\lambda_1-\frac{\sum \lambda_i c_i^2}{\sum c_i^2}
\end{align}
Now the expression is left with only the eigenvalues in $\pi^+$. Assume $k$ is the second smallest index of the eigenvalues in $\pi^+$, then 
\begin{align}\label{eq:3}
   D-c_1^2&\leq \frac{\epsilon D}{\lambda_1-\lambda_k}
.\end{align}
According to Gershgorin's Circle Theorem, $\lambda_1$ and $\lambda_2$ lie in the interval $[Q-m, Q+m]$, whereas the other eigenvalues are bounded by $[-m,m]$. The gap between the Rayleigh quotient and $\lambda_1$ now satisfies
$\epsilon\leq Q+m- R(H,y)$. It is easy to verify that the lower bound in Theorem \ref{thm:lambda_1} also applies to $R(H, \begin{bmatrix}
a&a&0
\end{bmatrix}^T)= R(H^+,\begin{bmatrix}
a&0
\end{bmatrix}^T)$.

When all the inequalities are combined, expression \eqref{eq:3} transforms into
\begin{align}
    D-c_1^2\leq 
    D \cdot \frac{m}{Q-m}
.\end{align}
Let $F_l$ be the set of vertices with minimum distance to either $v_1$ or $v_1'$ equals to $l$. Then the maximum possible value of the size of $F_l$ is $m^l$. This partition of the vertex set gives an upper bound on $D=2(1 + Q^{-2}|F_1|+Q^{-4}|F_2|+...+Q^{-2d}|F_d|)\leq \frac{2}{1-\frac{m}{Q^2}}$. Naturally, this requires $Q$ to be at least $\sqrt{m}$ to ensure that the series converges. By replacing the upper bound on $D$ in inequality \eqref{eq:4}, we obtain the lower bound on $\varphi_1(v_1)$ as stated in the theorem.

To determine the lower bound on $\varphi_2(v_1)^2$, we need a similar vector $\Tilde{y}$, but with alternating signs between vertex sets $N$ and $\sigma  N$. Let
\begin{equation}
  \tilde{y}_i =
    \begin{cases}
      Q^{-\min{\{\text{dist}(v_1,v_i)},{\text{dist}(v_1,v_i')}\}} & \text{if $v_i\in N$ and $\text{dist}(v_1,v_i)\neq\text{dist}(v_1,v_i')$}\\
      -Q^{-\min{\{\text{dist}(v_1,v_i)},{\text{dist}(v_1,v_i')}\}} & \text{if $v_i\in \sigma N$ and $\text{dist}(v_1,v_i)\neq\text{dist}(v_1,v_i')$}\\
      0 & \text{if $\text{dist}(v_1,v_i)=\text{dist}(v_1,v_i')$}
    \end{cases}       
.\end{equation}
And if it is written as a linear combination of the eigenvectors then coefficients $\tilde{c}_i=0$ for eigenvalues $\lambda_i$ in set $\pi^+$. Similarly, it follows that
\begin{equation}\label{eq:varphi2v1}
    \varphi_2(v_1)\geq \frac{1-\sqrt{D}\cdot\sqrt{\frac{m+1}{Q-m-1}}}{\sqrt{D}}\geq\sqrt{0.5-\frac{m}{2Q^2}}-\sqrt{\frac{m+1}{Q-m-1}}
\end{equation}
\end{proof}
\begin{xrem}
    These lower bounds depend solely on the ratio between $Q$ and the maximum degree of the graph. Regardless of the number of vertices, as long as $Q$ is significantly larger than the maximum degree, then the probability is close to 1.
\end{xrem}

\subsection{Time $t$ of Achieving Strong State Transfer}\label{se:t}
So far, we have demonstrated that given the maximum degree of the graph, one can ensure the probability of quantum state transfer between $v$ and $v'$ to be arbitrarily close to 1 by selecting a sufficiently large $Q$. For practical purposes, we also need to ensure that the time $t = \frac{\pi}{\lambda_1-\lambda_2}$, when $|\varphi_1(v_1)e^{it\lambda_1}+\varphi_2(v_1)e^{it\lambda_2}|$ is at its maximum, is not too large. To find a lower bound on $\lambda_1-\lambda_2$, we will first take a detour and study the eigenvectors using the reduced $2\times 2$ matrix first introduced in \cite{Tunneling_2012}.
\begin{lemma}\label{le:z}\cite{Tunneling_2012}
   $H$ is the adjacency matrix of a simple connected graph with potential $Q$ on vertices $u$, $v$. Let $\lambda$ be an eigenvalue of $H$ with corresponding eigenvector $\varphi$. If $\varphi(u)=\mu$ and $\varphi(v)=\nu$ then $\mu$ and $\nu$ satisfy
\begin{align}\label{mtx:Z}
    \begin{bmatrix}
        Z_{uu}(\lambda) & Z_{uv}(\lambda)\\
        Z_{uv}(\lambda) & Z_{vv}(\lambda)
    \end{bmatrix}
    \begin{bmatrix}
        \mu\\
        \nu
    \end{bmatrix}=(1-\frac{Q}{\lambda})
    \begin{bmatrix}
        \mu\\
        \nu
    \end{bmatrix}
\end{align}
 where $Z_{xy}(\lambda)=\sum_{P:x\rightarrow y} \frac{1}{\lambda^{|P|}}$ denotes the sum of $\frac{1}{\lambda^{|P|}}$ over all walks from $x$ to $y$; $|P|$ is the length of the walk. 
\end{lemma}
\begin{proof}
To construct a function $f: V(G)\rightarrow \mathbb{R}$ which we claim to be an eigenvector of $H$, let $f(u)=\mu$, $f(v)=\nu$, and $Hf(x)=\lambda f(x)$ for all $x\in V(G)\backslash \{u,v\}$. Since $f(x)=\frac{Af(x)}{\lambda}$ sums over all vertices $y$ adjacent to $x$ and divides $f(y)$ by $\lambda$, we can compute $f(x)$ by summing over all possible walks from the endpoints $u$ and $v$ to $x$. In particular, 
\begin{align}\label{function:f(x)}
    f(x)=\mu\cdot \sum_{P:x\rightarrow u} \frac{1}{\lambda^{|P|}}
    +\nu\cdot \sum_{P':x\rightarrow v} \frac{1}{\lambda^{|P'|}}
.\end{align}

In order for $f$ to be the actual eigenvector of $H$, $\lambda f(x) = Hf(x)$ should also be true at the endpoints:
\begin{align}\label{eq:lambdaf=hf}
    \lambda f(u)=Qf(u)+\sum_{x\sim u}f(x) \text{  ;  } \lambda f(v)=Qf(v)+\sum_{x\sim v}f(x)
.\end{align}
Divide both sides of equations \eqref{eq:lambdaf=hf} by $\lambda$, we get the following equality when $x$ is one of the endpoints
\begin{align}\label{eq:f(v)}
    f(x)=\frac{Q}{\lambda} f(x)
    +
    \mu\cdot \sum_{P:x\rightarrow x} \frac{1}{\lambda^{|P|}}
    +\nu\cdot \sum_{P':x\rightarrow v} \frac{1}{\lambda^{|P'|}}
,\end{align} which is exactly what the $2\times 2$ linear system in the lemma describes.
\end{proof}
% In conclusion, for two real numbers $\mu$ and $\nu$ to be the entries $\varphi(u)$, $\varphi(v)$ of the eigenvector corresponding to $\lambda$ and matrix $H$, they should satisfy the condition of forming an eigenvector of the matrix $Z$ described above. 
\begin{theorem}
    Given a graph with an involution and potential $Q$ on $v$ and $v'$. If the maximum degree is $m$ and the distance from $v$ to $v'$ is $d$, then
    \begin{align}
        \lambda_1-\lambda_2 >\frac{2}{(Q+m)^{d-1}}
.    \end{align}
\end{theorem}
\begin{proof}
    
When the graph has an involution $\sigma$ and $u=\sigma(v)$, taking a walk from $u$ to itself is equivalent to taking a walk from $v$ to $v$. Thus, the matrix in Lemma \ref{le:z} has eigenvectors $\begin{bmatrix}
    1&1
\end{bmatrix}^T$ and $\begin{bmatrix}
    1&-1
\end{bmatrix}^T$. 
According to Lemma \ref{le:H}$, \mu=\nu$ if $\lambda\in \pi^+$ and $\mu=-\nu$ if $\lambda\in \pi^-$. Therefore, $\lambda_1$ and $\lambda_2$ satisfy 
\begin{align}
    \lambda_1(Z_{vv}+ Z_{vv'})=\lambda_1-Q \text{  and  } \lambda_2(Z_{vv}- Z_{vv'})=\lambda_2-Q
.\end{align}
After comparing these two equations, we obtain
\begin{align}\label{eq:1.21}
    \lambda_1-\lambda_2=
    \sum_{P:v\rightarrow v} \left(\frac{1}{\lambda_1^{|P|-1}}
    -\frac{1}{\lambda_2^{|P|-1}}
    \right)+\sum_{P:v\rightarrow v'} \left(\frac{1}{\lambda_1^{|P|-1}}
    +\frac{1}{\lambda_2^{|P|-1}}\right)
.\end{align}
Move the first summation to the left-hand side, it becomes
\begin{align}
(\lambda_1-\lambda_2)\left(1+\frac{n_2}{\lambda_1\lambda_2}+\frac{n_3(\lambda_1+\lambda_2)}{\lambda_1\lambda_2}+...\right)>\lambda_1-\lambda_2
.\end{align}
Here $n_k$ is the number of walks from $v$ to itself of length $k$. 
The right hand side contains $\frac{1}{\lambda_1^{d-1}}$ and $\frac{1}{\lambda_2^{d-1}}$. Therefore, 
\begin{align}
    \lambda_1-\lambda_2 > \frac{1}{\lambda_1^{d-1}}+\frac{1}{\lambda_2^{d-1}}>\frac{2}{(Q+m)^{d(v,v')-1}}
.\end{align}
Consequently, the time it takes for $p(t)$ to be arbitrarily close to 1 is less than $\frac{\pi}{2}(Q+m)^{d-1}$.
\end{proof}

% \nocite{*}
\bibliographystyle{plain}
\bibliography{ref}

\end{document}